%% file: mmWaveCsBound.tex
\documentclass[conference]{IEEEtran}
\ifCLASSINFOpdf
 \usepackage[pdftex]{graphicx}
\else

\fi
\usepackage[aboveskip=1pt, belowskip=0pt]{subcaption}
\usepackage{amsmath}
\usepackage{amsthm}
\usepackage{amssymb}
\usepackage[shortlabels]{enumitem}
\usepackage{bbm}
\usepackage{scrextend}
\usepackage{mathtools}
\usepackage{extarrows}
\usepackage{accents}
\usepackage{optidef}
\usepackage{relsize}
\usepackage{mdframed}
\usepackage[most]{tcolorbox}

\tcbset{colback=white, colframe=black, 
        highlight math style= {enhanced, 
            colframe=black,colback=white,boxsep=0pt}
        }

\usepackage{cite}
\usepackage[utf8]{inputenc}

\DeclareMathOperator{\vect}{vec}
\DeclareMathOperator{\spark}{spark}
\DeclarePairedDelimiter\ceil{\lceil}{\rceil}

\theoremstyle{definition}
\newtheorem{definition}{Definition}[section]
 
\newtheoremstyle{mytheorem}
  {2pt}
  {2pt}
  {\itshape}
  {}
  {\itshape\bfseries}
  {.}
  {.5em}
  {\thmname{#1}\thmnumber{ #2} \thmnote{ {\the\thm@notefont(#3)}}}

\makeatother
\theoremstyle{mytheorem}

\theoremstyle{remark}
\newtheorem*{remark}{Remark}
\newtheorem{example}{Example}

\newtheorem{theorem}{Theorem}

\newtheorem{lemma}[theorem]{Lemma}

\makeatletter
\renewcommand*\env@matrix[1][*\c@MaxMatrixCols c]{%
  \hskip -\arraycolsep
  \let\@ifnextchar\new@ifnextchar
  \array{#1}}
\makeatother

\newcommand{\norm}[1]{\left\lVert#1\right\rVert}

\DeclarePairedDelimiter\abs{\lvert}{\rvert}%

\makeatletter
\let\oldabs\abs
\def\abs{\@ifstar{\oldabs}{\oldabs*}}

\newcommand{\subparagraph}{}
\usepackage{titlesec}
\titlespacing*{\section}{0pt}{0.3\baselineskip}{0.1\baselineskip}
\titlespacing*{\subsection}{0pt}{0.2\baselineskip}{0.1\baselineskip}

\makeatletter
\newenvironment{proof*}[1][\proofname]{\par
  \pushQED{\qed}%
  \normalfont \partopsep=\z@skip \topsep=\z@skip
  \trivlist
  \item[\hskip\labelsep
        \itshape
    #1\@addpunct{.}]\ignorespaces
}{%
  \popQED\endtrivlist\@endpefalse
}
\makeatother

\def\longversion{}

\begin{document}

\ifdefined\longversion
\else
\setlength{\abovedisplayskip}{3pt}
\setlength{\belowdisplayskip}{3pt}
\setlength{\textfloatsep}{6pt} 
\setlength{\floatsep}{1pt} 
\setlength{\skip\footins}{0.2cm} 
\setlength{\dbltextfloatsep}{5pt}
\setlength{\abovecaptionskip}{2pt plus 0pt minus 0pt} 
\setlength{\belowcaptionskip}{0pt plus 0pt minus 0pt} 
\fi

\title{How Long to Estimate Sparse MIMO Channels}

\author{\IEEEauthorblockN{Yahia Shabara, C. Emre Koksal and Eylem Ekici}
\IEEEauthorblockA{Dept. of ECE,
The Ohio State University, Columbus, OH 43210\\
Email: \{shabara.1, koksal.2 , ekici.2\}@osu.edu}
}

\maketitle

\thispagestyle{plain}
\pagestyle{plain}



\IEEEpeerreviewmaketitle

\input{text/abstract}

\section{Introduction}
\input{text/Intro}

\section{System Model}
\label{sec:sysModel}
\input{text/sysModel}

\section{Problem Formulation}
\input{text/PF_intro}
\subsection{Compressed Sensing Background}
\input{text/CS}
\subsection{The Problem}
\input{text/theProblem}

\section{Lower Measurement Bound}
\input{text/lowerBound}

\section{Conclusion}
\input{text/conclusion}

\appendices

\ifdefined\longversion
\else
\renewcommand{\thesectiondis}[2]{\Alph{section}:}
\fi




\ifdefined\longversion
\section{Proof of Lemma \ref{lemma:RIP_kronecker}}
\input{text/RIP_KP}
\label{append:RIP_KP}

\section{Proof Of Lemma \ref{lemma:BSC_lb}}
\input{text/Appendix_V}
\label{append:lemma_BSC_lb}

\section{Proof Of Lemma \ref{lemma:measurement_upperbound}}
\input{text/Appendix_III}
\label{append:lemma_proof_upperBound}

\section{Spark of the Kronecker Product}
\input{text/spark_kron}
\label{append:spark_kron}
\fi


\bibliographystyle{IEEEtran}
\linespread{1}
\bibliography{mmWaveCsBound}

%
%

\end{document}

%% file: text/abstract.tex
\begin{abstract}
Large MIMO transceivers are integral components of next-generation wireless networks.
However, for such systems to be practical, their channel estimation process needs to be fast and reliable.
Although several solutions for fast estimation of sparse channels do exist, there is still a gap in understanding the fundamental limits governing this problem.
Specifically, we need to better understand the lower bound on the number of measurements under which accurate channel estimates can be obtained.
This work bridges that knowledge gap by deriving a tight asymptotic lower bound on the number of measurements.
This not only helps develop a better understanding for the sparse MIMO channel estimation problem, but it also provides a benchmark for evaluating current and future solutions.
\end{abstract}

%% file: text/intro.tex
Through the use of a large number of antennas, wireless transceivers can focus their signal transmission and/or reception through very narrow angular directions \cite{tse2005fundamentals}.
This helps increase the channel capacity in two main ways. First, it improves the spatial multiplexing capability of transceivers, which allows  simultaneously serving multiple users while keeping cross interference low.
Second, it allows more signal power to be propagated from a transmitter (TX) to a receiver (RX).
For the latter reason, large MIMO transceivers have emerged as the prominent solution to solve the severe path loss problem in millimeter-wave (mmWave) systems \cite{shabara2019beam, fan2018angle}.

\ifdefined\longversion
The main challenge of large MIMO, however, is that the channel estimation process can be complex \cite{ding2018dictionary}.
This is a byproduct of having channel matrices with large dimensions.
Moreover, both initial and running costs (i.e., cost of hardware and power consumption, respectively) of such devices are high.
To minimize these costs, the architectural design of large MIMO transceivers have deviated from the traditional fully-digital design towards analog or hybrid transceivers.
While these alternative architectures solve the cost problem, they exacerbate the channel estimation overhead.
This is because such alternative transceiver designs are less-flexible than the fully-digital ones.
For example, an analog transceiver can obtain only one independent measurement at a time, unlike a digital transceiver that obtains as many independent measurements as the number of antennas at RX.
\else
The main challenge of large MIMO, however, is that the channel estimation process can be complex \cite{ding2018dictionary} since channel matrices have large dimensions.
This problem is further exacerbated by the practically-viable transceiver designs used to overcome the cost and power consumption problems attributed with the traditional fully-digital transceiver architectures.
\fi

Reducing the number of channel measurements is thus one of the main challenges facing large MIMO implementations.
This problem has largely been tackled as an application of \textbf{Compressed Sensing} (CS) \cite{choi2017compressed,donoho2006compressed}, which relies on channel \textit{sparsity} as a key enabler for reducing the number of measurements\footnote{Sparsity here means that the number of signal propagation paths is small compared to the number of TX and RX antennas (e.g., mmWave channels).}.
The closest effort to understanding how changing the number of measurements affects the quality of channel estimates, to the best of our knowledge, is \cite{Alkhateeb_2015_HowManyMeasurements}, where computer simulations were conducted to measure the quality of channel estimates as the number of measurements increases.
Nonetheless, there is still a gap in the current literature in understanding the \textit{lower bound} on the number of necessary measurements needed for accurate channel recovery.
To the best of our knowledge, the tightest known bound scales as
$\Omega\left(k \log \frac{n_t n_r}{k}\right)$ \cite{CS4Wireless_TipsAndTricks,alkhateeb2014channel}, where $k$ is the channel sparsity level and $n_t$ and $n_r$ are the numbers of antennas at TX and RX, respectively.
This bound, however, is a naive application of the CS bound for recovery of sparse vectors of length $n = n_t n_r$ and $k$ non-zero values.
In fact, the nature of the channel estimation problem poses limitations on how measurements are obtained, as opposed to the standard CS problem.
Thus, more attention needs to be paid when deriving measurement lower bounds.
In this paper, we show that the aforementioned bound is too loose, and we provide a tighter lower bound which has order of $\Omega\left(k^2 \log\left(\frac{n_t}{k}\right) \log\left(\frac{n_r}{k}\right)\right)$.
We argue the tightness of this bound by showing that, under a mild constraint on the channel sparsity level, there exists a solution with a number of measurements upper bounded as $O(k^2 \log\left(\frac{n_t}{k}\right) \log\left(\frac{n_r}{k}\right))$.

\ifdefined\longversion
\textit{Notations:} Let $x$ be a scalar quantity, $\boldsymbol{x}$ be a vector and $\boldsymbol{X}$ be a matrix. The conjugate of $\boldsymbol{X}$ is $\boldsymbol{X}^{\ast}$, its transpose is $\boldsymbol{X}^T$ and its hermition (i.e., conjugate transpose) is $\boldsymbol{X}^H$. Let $\norm{\boldsymbol{x}}_p$ denote the $p^{\text{th}}$ norm of $\boldsymbol{x}$. If the subscript $p$ is dropped, then $\norm{\boldsymbol{x}}$ denotes the Euclidean norm, $\norm{\boldsymbol{x}}_2$.
Define the operator $\vect \left( \boldsymbol{X} \right)$ to be the stacking of all the columns of $\boldsymbol{X}$ to form one vector as follows:
If $\boldsymbol{X}$ has columns $\boldsymbol{x_i}$ for $i = 1, \dots, n$, then
$\vect \left( \boldsymbol{X} \right) = \begin{pmatrix}
\boldsymbol{x_1}^T &
\boldsymbol{x_2}^T &
\dots &
\boldsymbol{x_n}^T
\end{pmatrix}^T$.
We denote by $\otimes$ the Kronecker product.
Finally, we use: (i) $\Omega\left(\cdot\right)$ to denote the Big Omega notation, i.e., the asymptotic lower bound\footnote{We say that $f(n) \in \Omega\left( g(n) \right)$ (or loosely, $f(n) {=} \Omega\left( g(n) \right)$) if there exists a constant $c > 0$, and $n_0 \in \mathbb{N}$ such that $f(n) \geq c g(n)$, for all $n {\geq} n_0$.},
(ii) $O\left(\cdot\right)$ to denote the Big O notation, i.e., the asymptotic upper bound\footnote{We say that $f(n) \in O\left( g(n) \right)$ (or loosely $f(n) {=} O\left( g(n) \right)$) if there exists a constant $c > 0$ and $n_0 \in \mathbb{N}$ such that $f(n) \leq c g(n)$, for all $n {\geq} n_0$.}, and (iii) we say that $f(n) \in \Theta\left(g(n)\right)$ if both $f(n) \in \Omega(g(n))$ and $f(n) \in O(g(n))$.
\else
\textit{Notations:} Let $x$ be a scalar quantity, $\boldsymbol{x}$ be a vector and $\boldsymbol{X}$ be a matrix. The Conjugate, Transpose and Hermition of $\boldsymbol{X}$ are denoted by $\boldsymbol{X}^{\ast}$, $\boldsymbol{X}^T$ and $\boldsymbol{X}^H$, respectively. The $p^{\text{th}}$ norm is denoted by $\norm{\boldsymbol{x}}_p$ (if the subscript $p$ is dropped, then assume $p{=}2$).
Denote by $\vect \left( \boldsymbol{X} \right)$ the vectorization of columns of $\boldsymbol{X}$, and denote by $\otimes$ the Kronecker product.
Finally, we use: (i) $\Omega\left(\cdot\right)$ to denote the Big Omega notation, i.e., the asymptotic lower bound\footnote{We say that $f(n) \in \Omega\left( g(n) \right)$ (or loosely, $f(n) {=} \Omega\left( g(n) \right)$) if there exists a constant $c > 0$, and $n_0 \in \mathbb{N}$ such that $f(n) \geq c g(n)$, for all $n {\geq} n_0$.},
(ii) $O\left(\cdot\right)$ to denote the Big O notation, i.e., the asymptotic upper bound\footnote{We say that $f(n) \in O\left( g(n) \right)$ (or loosely $f(n) {=} O\left( g(n) \right)$) if there exists a constant $c > 0$ and $n_0 \in \mathbb{N}$ such that $f(n) \leq c g(n)$, for all $n {\geq} n_0$.}, and (iii) we say that $f(n) \in \Theta\left(g(n)\right)$ if both $f(n) \in \Omega(g(n))$ and $f(n) \in O(g(n))$.
\fi

%% file: text/sysModel.tex
Consider a single-tap, block-fading, \textbf{\textit{sparse}} MIMO channel between a TX and RX equipped with $n_t$ and $n_r$ antennas, respectively.
Antennas at TX and RX form Uniform Linear Arrays (ULA), with \textit{normalized} antenna spacing of $\Delta_t$ and $\Delta_r$, respectively.
The normalization is with respect to the carrier wavelength, denoted by $\lambda_c$.
We consider \textit{analog} transceiver architectures at both TX and RX.
That is, only one RF chain exists per transceiver, and all antennas are connected to this RF chain through phase-shifters and variable-gain amplifiers.

Let the maximum number of \textit{resolvable} signal propagation paths in the channel be denoted by $k$. Recall that we consider sparse channels.
By the sparsity assumption \cite{ding2018dictionary, gao2015asymptotic, masood2014efficient, sun2017millimeter, choi2017compressed, ma2019deep}, only a few signal propagation paths exist, where $k \ll n_t, n_r$.
Note that a wireless transceiver may not be able to resolve multiple channel paths if they are spatially close.
However, as the number of antennas increases, the transceiver's ability to resolve more paths also increases due to its ability to form narrower antenna beams.
This means that $k$ increases with $n$.
However, the ratio $\frac{k}{n}$ decreases as $n$ increases.
We assume that  $n_t, n_r \geq k^{1+\epsilon}$, for some $\epsilon {>} 0$, which reflects the ability of transceivers to resolve more channel paths as their number of antennas increases.
For each propagation path $p$, let $\alpha_p$ be its path-gain, $\theta_p$ be its Angle of Departure (AoD) at TX, $\phi_p$ be its Angle of Arrival (AoA) at RX, and $\rho_p$ be its path length.
The baseband path gain, $\alpha_p^b$, is given by
\ifdefined\longversion
\begin{equation}
\alpha_p^b = \alpha_p \sqrt{n_t n_r} \exp^{-j \frac{2 \pi \rho_p}{\lambda_c}}.
\end{equation}
\else
{\footnotesize
$\alpha_p^b = \alpha_p \sqrt{n_t n_r} \exp^{-j \frac{2 \pi \rho_p}{\lambda_c}}$.}%
\fi

Let $\boldsymbol{Q} \in \mathbb{C}^{n_r \times n_t}$ denote the channel matrix, where $q_{i,j}$, the element at row $i$ and column $j$ in $\boldsymbol{Q}$, is the channel gain between the $j^{\text{th}}$ TX antenna and the $i^{\text{th}}$ RX antenna.
Let us denote the path-loss by $\mu$. Then, we can write $\boldsymbol{Q}$ as
\begin{equation}
\boldsymbol{Q} = \sum_{p = 1}^k \frac{\alpha_p^b}{\mu} \boldsymbol{e_r}(\Omega_{r,p}) \boldsymbol{e}^H_{\boldsymbol{t}}(\Omega_{t,p}),
\end{equation}
where $\boldsymbol{e_t}(\Omega)$ and $\boldsymbol{e_r}(\Omega)$ are the transmit and receive signal spatial signatures, at angular cosine $\Omega$ \cite[Chapter 7]{tse2005fundamentals}.
\ifdefined\longversion
We define $\boldsymbol{e_i}(\Omega)$ as:
\begin{equation}
\small
\boldsymbol{e_i}(\Omega) = \frac{1}{\sqrt{n_i}}
\begin{pmatrix}
1 \\
\exp^{-j2\pi           \Delta_i \Omega}\\
\exp^{-j2\pi  2        \Delta_i \Omega}\\
\vdots                                 \\
\exp^{-j2\pi (n_i - 1) \Delta_i \Omega}
\end{pmatrix},
\quad \quad  i \in \{t,r\}.
\end{equation}
\fi
The channel $\boldsymbol{Q}$, in this form, is not sparse. However, it can be represented in a sparse form using a simple change of basis:
\begin{equation}
\label{eqn:SparsifyingTheChannel}
\boldsymbol{Q^a} = \boldsymbol{U}_{\boldsymbol{r}}^H \boldsymbol{Q} \boldsymbol{U_t},
\end{equation}
where $\boldsymbol{Q^a}$ is known as the ``\textit{angular channel}'' and is sparse.
The matrices $\boldsymbol{U_t}$ and $\boldsymbol{U_r}$ are Discrete Fourier Transform matrices whose columns represent an orthonormal basis for the transmit and receive signal spaces, and are defined as:
\begin{equation*}\small
\begingroup 
\setlength\arraycolsep{2pt}
\boldsymbol{U_i} =
\begin{pmatrix}
\boldsymbol{e_i}\left( 0 \right) & \boldsymbol{e_i}\left(\frac{1}{L_i}\right) & \boldsymbol{e_i}\left(\frac{2}{L_i}\right) & \dots  & \boldsymbol{e_i}\left(\frac{n_i - 1}{L_i}\right)
\end{pmatrix},
\quad i \in \{t,r\}.
\endgroup
\end{equation*}


When transmitting a symbol $\zeta$, the TX uses a precoder vector $\boldsymbol{f} \in \mathbb{C}^{n_t}$ while RX uses a combiner vector $\boldsymbol{w} \in \mathbb{C}^{n_r}$. The received symbol at RX is thus given by:
\begin{equation}
y_{i,j} = \boldsymbol{w}_i^H \boldsymbol{Q} \boldsymbol{f}_j \zeta + \boldsymbol{w}_i^H \boldsymbol{n_{i,j}},
\end{equation}
where $y_{i,j}$ denotes the received symbol (i.e., measurement result), $\boldsymbol{w}_i$ denotes the $i^{\text{th}}$ receive combiner and $\boldsymbol{f}_j$, the $j^{\text{th}}$ transmit precoder.
Assume, for simplicity, that $\zeta=1$.
Let the number of rx-combiners be $m_r$ and the number of tx-precoders be $m_t$.
Then, the total number of measurements we can obtain using all combinations of $\boldsymbol{f}_j$ and $\boldsymbol{w}_i$ is $m = m_t {\times} m_r$.
We can also write the measurement equations for all precoders and combiners more compactly as:
\begin{equation}
\label{eq:nonAdaptiveMatrixForm}
\boldsymbol{Y} = \boldsymbol{W}^H \boldsymbol{Q} \boldsymbol{F} + \boldsymbol{N},
\end{equation}
where $y_{i,j}$ is the element at row $i$ and column $j$ of $\boldsymbol{Y}$. $\boldsymbol{W}$ and $\boldsymbol{F}$ are defined as:
\small{
\ifdefined\longversion
\begin{align}
\boldsymbol{W} &\triangleq \begin{pmatrix} \boldsymbol{w}_1 & \boldsymbol{w}_2 & \dots & \boldsymbol{w}_{m_r} \end{pmatrix}, \\
\boldsymbol{F} &\triangleq \begin{pmatrix} \boldsymbol{f}_1 & \boldsymbol{f}_2 & \dots & \boldsymbol{f}_{m_t}
\end{pmatrix}
\end{align}
\else
\begin{equation}
\begingroup 
\setlength\arraycolsep{2pt}
\boldsymbol{W} \triangleq \begin{pmatrix} \boldsymbol{w}_1 & \boldsymbol{w}_2 & \dots & \boldsymbol{w}_{m_r} \end{pmatrix}, \quad
\boldsymbol{F} \triangleq \begin{pmatrix} \boldsymbol{f}_1 & \boldsymbol{f}_2 & \dots & \boldsymbol{f}_{m_t}
\end{pmatrix}
\endgroup
\end{equation}
\fi
}%
\normalsize
The channel estimation problem, i.e., figuring out what the matrix $\boldsymbol{Q}$ is, can be broken down into determining the best set of precoders $\boldsymbol{f}_j$ and combiners $\boldsymbol{w}_i$ using which we can accurately recover $\boldsymbol{Q}$. To speed up the estimation process, the smallest sets of those $\boldsymbol{f}_j$'s and $\boldsymbol{w}_i$'s should be used.
In this paper, we do not provide a specific design for such precoders and combiners, but we seek to find a ``tight'' lower bound on the number of measurements using which $\boldsymbol{Q}$ can be recovered.

\ifdefined\longversion
\textbf{Special Cases:}
Suppose the number of TX antennas $n_t {=} 1$. In such case, the channel is Single-Input-Multiple-Output (SIMO), and the channel matrix $\boldsymbol{Q}$ becomes a vector $\boldsymbol{q}$. The precoders at TX also fall back to just a scalar quantity; $f = 1$. Thus, we can rewrite the measurement equation (Eq. (\ref{eq:nonAdaptiveMatrixForm})) as:
\begin{equation}
\label{eqn:SIMO_measurements}
\boldsymbol{y} = \boldsymbol{W}^H \boldsymbol{q} + \boldsymbol{n}
\end{equation}
Similarly, if we have a MISO channel, i.e., $n_r {=} 1$, we have the following measurement equation:
\begin{equation}
\label{eqn:MISO_measurements}
\boldsymbol{y} = \boldsymbol{F}^H \boldsymbol{q} + \boldsymbol{n}
\end{equation}
\else
\textbf{Special Cases:}
Consider the special cases in which either (i) $n_t = 1$ or (ii) $n_r = 1$.
In the former case, the channel is Single-Input-Multiple-Output (SIMO), while in the latter it is MISO. In both cases, the channel matrix $\boldsymbol{Q}$ becomes a vector $\boldsymbol{q}$. The tx-precoders in the MISO scenario, fall back to just a scalar quantity; $f = 1$, while in SIMO, rx-combiners fall back to $w = 1$.
Thus, we can rewrite the measurement equation (Eq. (\ref{eq:nonAdaptiveMatrixForm})) for SIMO and MISO, respectively, as follows
\begin{equation}
\label{eqn:SIMO_MISO_measurements}
\boldsymbol{y} = \boldsymbol{W}^H \boldsymbol{q} + \boldsymbol{n}, \quad \quad \quad
\boldsymbol{y} = \boldsymbol{F}^H \boldsymbol{q} + \boldsymbol{n}.
\end{equation}
\fi

%% file: text/PF_intro.tex
In this section, we will provide a brief overview of compressed sensing (CS). Then, we will formulate the problem of channel estimation as a CS problem.
To that end, we will reshape the measurement equation given in Eq. (\ref{eq:nonAdaptiveMatrixForm}) to be in the form
$\boldsymbol{y_v} {=} \boldsymbol{G_v} \boldsymbol{q^a_v} {+} \boldsymbol{n_v}$,
which conforms with the traditional compressed sensing problem, as will be shown in Eq. (\ref{eq:CS_stdForm_withNoise}) below.
Here, $\boldsymbol{q^a_v}$ is sparse and has dimensions $n_r n_t {\times} 1$.

%% file: text/CS.tex
Compressed sensing is a signal processing technique \cite{donoho2006compressed} that allows the reconstruction of a signal $\boldsymbol{x} = \left( x_i \right)_{i = 1}^n$ from a small number of samples \textit{given that} $\boldsymbol{x}$ is either: (i) sparse, or (ii) can be represented in a sparse form, using a linear transformation $\boldsymbol{U}$ such that $\boldsymbol{x} = \boldsymbol{U} \boldsymbol{s}$ where $\boldsymbol{s}$ is sparse.
Let the number of measurements be denoted by $m$ where $m < n$ and $m,n \in \mathbb{N}$.
Each measurement of $\boldsymbol{x}$ is a linear combination of its components $x_i$.
Such measurements are dictated by the sensing matrix $\boldsymbol{G}$ and are given by
\ifdefined\longversion
\begin{equation}
\label{eq:CS_stdForm}
\boldsymbol{y} = \boldsymbol{G} \boldsymbol{x},
\end{equation}
where $\boldsymbol{y}$ denotes the $m {\times} 1$ measurement vector.
Eq. (\ref{eq:CS_stdForm}) represents an under-determined system of linear equations (since $m < n$).
\else
$\boldsymbol{y} = \boldsymbol{G} \boldsymbol{x}$,
where $\boldsymbol{y}$ denotes the $m {\times} 1$ measurement vector. The matrix equation $\boldsymbol{y} = \boldsymbol{G} \boldsymbol{x}$ represents an under-determined system of linear equations (since $m < n$).
\fi
In other words, we have fewer equations than the number of unknowns we want to solve for.
While, in general, an infinite number of solutions exist, the sparsity of $\boldsymbol{x}$ allows for perfect signal reconstruction from $\boldsymbol{y}$ \textit{given that} certain conditions are satisfied, among which, is a lower bound on the ``\textit{spark}'' of the sensing matrix.
\begin{definition}
The spark of a given matrix $\boldsymbol{G}$ is the smallest number of its linearly dependent columns.
\end{definition}
\begin{theorem} [Corollary 1 of \cite{donoho2003optimally}]
\label{thm:spark}
For any vector $\boldsymbol{y} \in \mathbb{R}^m$, there exits at most one vector $\boldsymbol{q^a} \in \mathbb{R}^n$ with $\norm{\boldsymbol{q^a}}_0 = k$ such that $\boldsymbol{y} = \boldsymbol{G} \boldsymbol{q^a}$ if and only if $\spark(\boldsymbol{G}) > 2k$.
\end{theorem}
Theorem \ref{thm:spark} provides a mathematical guarantee on the exact recovery of $k-$sparse vectors using $m$ linear measurements.
An immediate bound on the number of measurements, $m$, we get from Theorem \ref{thm:spark} is
\ifdefined\longversion
  \begin{equation}
    m \geq 2k.
  \end{equation}
\else
  $m \geq 2k.$
\fi
\ifdefined\longversion
The spark lower bound on the matrix $\boldsymbol{G}$ works well under noise-free measurements.
In practice, however, measurements are corrupted with an error vector $\boldsymbol{n}$, i.e.,
\else
The lower bound on the spark of $\boldsymbol{G}$ works well under noise-free measurements, but in practice, measurements get corrupted with a vector $\boldsymbol{n}$, i.e.,
\fi
\begin{equation}
\label{eq:CS_stdForm_withNoise}
\boldsymbol{y} = \boldsymbol{G} \boldsymbol{x} + \boldsymbol{n}.
\end{equation}
It is necessary to guarantee that the measurement process is not adversely affected by such errors in a significant way.
This calls for alternative, stricter requirements on sensing matrices to guarantee ``good'' sparse recovery.
Mathematically, we need to design the sensing matrix such that the energy in the measured signal is preserved.
This is quantified using the \textbf{Restricted Isometry Property (RIP)}.
The RIP property guarantees that the distance between any pair of $k-$sparse vectors is not significantly changed under the measurement process.
This RIP property is defined as follows:
\begin{definition}
A matrix $\boldsymbol{G}$ satisfies the restricted isometry property (RIP) of order $k$ if there exists a constant $\delta_k \in (0,1)$ such that for all vectors $\boldsymbol{q^a}$, with $\norm{\boldsymbol{q^a}}_0 \leq k$, we have
\begin{equation}\label{eq:RIP}
(1 - \delta_k) \norm{\boldsymbol{q^a}}_2^2 \leq \norm{\boldsymbol{G}\boldsymbol{q^a}}_2^2 \leq (1 + \delta_k) \norm{\boldsymbol{q^a}}_2^2.
\end{equation}
\end{definition}
The smallest $\delta_k$ which satisfies Eq. (\ref{eq:RIP}) is called the ``$k-$restricted isometry constant''.
Note that in general, a matrix $\tilde{\boldsymbol{G}}$ does not necessarily result in $\small ||\tilde{\boldsymbol{G}}\boldsymbol{q^a}||^2$ that is symmetric about $1$. However, a simple scaling of $\tilde{\boldsymbol{G}}$ results in $\boldsymbol{G}$ such that the tightest bounds of $\small \norm{\boldsymbol{G}\boldsymbol{q^a}}^2$ in Eq. (\ref{eq:RIP}) are symmetric \cite{eldar2012compressed}.
From now on, we will only consider matrices whose bounds are symmetric as shown in Eq. (\ref{eq:RIP}). 

The following theorem provides a necessary condition for $m {\times} n$ matrices that satisfy the RIP property with $\delta_{k} {\in} \left(0,1\right)$.
\begin{theorem}[Theorem 3.5 of \cite{davenport2010random}]
\label{thm:davenPortBnd}
Let $\boldsymbol{G}$ be an $m {\times} n$ matrix that satisfies RIP of order $k$ with constant $\delta_{k} {\in} \left(0,1\right)$. Then,
\begin{equation}
\label{eq:CS_boundRIP}
m \geq c_{\delta} k \log\left( \frac{n}{k} \right)
\end{equation}
where $c_{\delta} = \frac{0.18}{\log\left( \sqrt{\frac{1+\delta}{1-\delta}}+1 \right)}$, is a function of $\delta$ only.
\end{theorem}
Theorem \ref{thm:davenPortBnd} demonstrates the popular asymptotic measurement bound:
\ifdefined\longversion
  \begin{equation}
    \label{eqn:popular_bound}
    m = \Omega\left(k \log \frac{n}{k}\right).
  \end{equation}
\else
  $m = \Omega\left(k \log \frac{n}{k}\right)$.
\fi
Next, we will formulate the MIMO channel estimation as a compressed sensing problem.

%% file: text/theProblem.tex
Recall from Eq. (\ref{eq:nonAdaptiveMatrixForm}) that channel measurements take the form
\ifdefined\longversion
\begin{equation*}
\boldsymbol{Y} = \boldsymbol{W}^H \boldsymbol{Q} \boldsymbol{F} + \boldsymbol{N}.
\end{equation*}
\else
$\boldsymbol{Y} {=} \boldsymbol{W}^H \boldsymbol{Q} \boldsymbol{F} + \boldsymbol{N}$.
\fi
This is not the standard form of a noisy CS problem (see Eq. (\ref{eq:CS_stdForm_withNoise})).
Thus, it cannot readily be solved using compressed sensing.
\ifdefined\longversion
To put this equation in a CS problem form, let us ``\textit{vectorize}'' its left and right hand sides as follows:
\begin{itemize}
\item Let $\boldsymbol{y_v} = \vect \left( \boldsymbol{Y} \right)$
\item Let $\boldsymbol{n_v} = \vect \left( \boldsymbol{N} \right)$
\item And by the properties of \textit{vectorization} \cite{dhrymes1978mathematics}, we have
\footnotesize{
  \begin{align}
    \vect \left( \boldsymbol{W}^H  \boldsymbol{Q} \boldsymbol{F}\right)
    & = \left( \boldsymbol{F}^T \otimes \boldsymbol{W}^H \right) \vect \left( \boldsymbol{Q}
    \right) \\
    & = \left( \boldsymbol{F}^T \otimes \boldsymbol{W}^H \right) \vect \left( \boldsymbol{U}
    _{\boldsymbol{r}} \boldsymbol{Q^a} \boldsymbol{U}^H_{\boldsymbol{t}} \right) \\
    & = \left( \boldsymbol{F}^T \otimes \boldsymbol{W}^H \right)
    \left( \boldsymbol{U}_{\boldsymbol{t}}^{\ast} \otimes \boldsymbol{U}_{\boldsymbol{r}}
     \right)
    \vect \left( \boldsymbol{Q^a} \right)\\
    & = \left( \boldsymbol{F}^T \otimes \boldsymbol{W}^H \right)
    \left( \boldsymbol{U}_{\boldsymbol{t}}^{\ast} \otimes \boldsymbol{U}_{\boldsymbol{r}}
    \right)
    \boldsymbol{q^a_v}\\
    & = \left( \left( \boldsymbol{F}^T \boldsymbol{U}^{\ast}_{\boldsymbol{t}} \right) \otimes
    \left( \boldsymbol{W}^H \boldsymbol{U_r} \right) \right) \boldsymbol{q^a_v}\\
    & = \left( \left( \boldsymbol{F}^H \boldsymbol{U_t} \right)^{\ast} \otimes \left(
    \boldsymbol{W}^H \boldsymbol{U_r} \right) \right) \boldsymbol{q^a_v}
  \end{align}}%
  \normalsize
\end{itemize}
\else
To put this equation in a CS problem form, let us ``\textit{vectorize}'' its left and right hand sides as follows:
(i)   let $\boldsymbol{y_v}   {\triangleq} \vect \left( \boldsymbol{Y}   \right)$,
(ii)      $\boldsymbol{n_v}   {\triangleq} \vect \left( \boldsymbol{N}   \right)$,
(iii)     $\boldsymbol{q^a_v} {\triangleq} \vect \left( \boldsymbol{Q^a} \right)$, and finally
(iv) by properties of \textit{vectorization} \cite{dhrymes1978mathematics}, we have
\par\nobreak
\vspace{-9pt}
{\footnotesize
  \begin{align}
    \vect ( \boldsymbol{W}^H  \boldsymbol{Q} \boldsymbol{F})
    & = ( \boldsymbol{F}^T \otimes \boldsymbol{W}^H ) \vect ( \boldsymbol{Q}) \\
    & = ( \boldsymbol{F}^T \otimes \boldsymbol{W}^H )
    \left( \boldsymbol{U}_{\boldsymbol{t}}^{\ast} \otimes \boldsymbol{U}_{\boldsymbol{r}}
     \right)
    \vect \left( \boldsymbol{Q^a} \right)\\
    & = \left( \left( \boldsymbol{F}^H \boldsymbol{U_t} \right)^{\ast} \otimes \left(
    \boldsymbol{W}^H \boldsymbol{U_r} \right) \right) \boldsymbol{q^a_v}
  \end{align}}%
\fi
Thus, we can rewrite the measurement equation in (\ref{eq:nonAdaptiveMatrixForm}) as
\small{
\begin{align}
\boldsymbol{y_v} &= \boldsymbol{G_v} \boldsymbol{q^a_v}  + \boldsymbol{n_v}, \label{eq:sensingVectEq} \\
\textit{where} \quad \boldsymbol{G_v} &= \left( \boldsymbol{F}^H \boldsymbol{U_t} \right)^{\ast} \otimes \left( \boldsymbol{W}^H \boldsymbol{U_r} \right) \label{eq:sensingVect}
\end{align}
}%
\normalsize
is the sensing matrix, with dimensions $m_t m_r {\times} n_t n_r$, while $\boldsymbol{y_v}$ has dimensions $m_t m_r {\times} 1$ and $\boldsymbol{q^a_v}$ has dimensions $n_t n_r {\times} 1$.
This form of the problem allows us to employ CS sparse recovery techniques to estimate $\boldsymbol{q^a_v}$ from $\boldsymbol{y_v}$.

%% file: text/lowerBound.tex
We are interested in sensing matrices that preserve the \textbf{\textit{distance}} between two different channels $\boldsymbol{q^a_{v1}}$ and $\boldsymbol{q^a_{v2}}$. This distance is the norm of $\boldsymbol{q^a_{v1}} - \boldsymbol{q^a_{v2}}$, which has a sparsity level of $2k$ (recall that the maximum number of channel paths is $k$).
Thus, to be able to accurately estimate $\boldsymbol{q^a_v}$, we need the sensing matrix $\boldsymbol{G_v}$ to satisfy the RIP property of order $2k$ with some RIP constant $\delta_{2k} \in (0,1)$.
At sparsity level of $2k$, Theorem \ref{thm:davenPortBnd} shows that the recovery of a sparse vector with dimensions $n {=} n_t n_r$ requires a number of measurements, $m$, lower bounded as
\ifdefined\longversion
  \begin{align}\small
    m
    \small &\geq c_{\delta} (2k) \log \left( \frac{n_t n_r}{(2k)} \right) \\
    \small &= 2 c_{\delta} k \left( \log \left( \frac{n_t}{\sqrt{2k}} \right) + \log \left(
    \frac{n_r}{\sqrt{2k}} \right) \right). \label{eq:CSWrongBound}
\end{align}
\else
    $\small m \geq c_{\delta} (2k) \log \left( \frac{n_t n_r}{(2k)} \right)
    = 2 c_{\delta} k \left( \log \left( \frac{n_t}{\sqrt{2k}} \right)
    + \log \left( \frac{n_r}{\sqrt{2k}} \right) \right)$.
\fi
This demonstrates the popular $m = \Omega\left(k \log\left(\frac{n_r \times  n_t}{k}\right)\right)$ lower bound for sparse channel estimation.
Although this bound is valid, it is in fact too loose since it assumes that arbitrary constructions of $\boldsymbol{G}_v$ are possible.
This, however, is not the case for sparse MIMO channel estimation since $\boldsymbol{G}_v$ takes a special, Kronecker product form, as derived in Eq. (\ref{eq:sensingVect}).

\ifdefined\longversion
\begin{figure*}
\centering
\begin{subfigure}{.49\textwidth}
\centering
\includegraphics[height=3.5cm, width=0.9\linewidth]{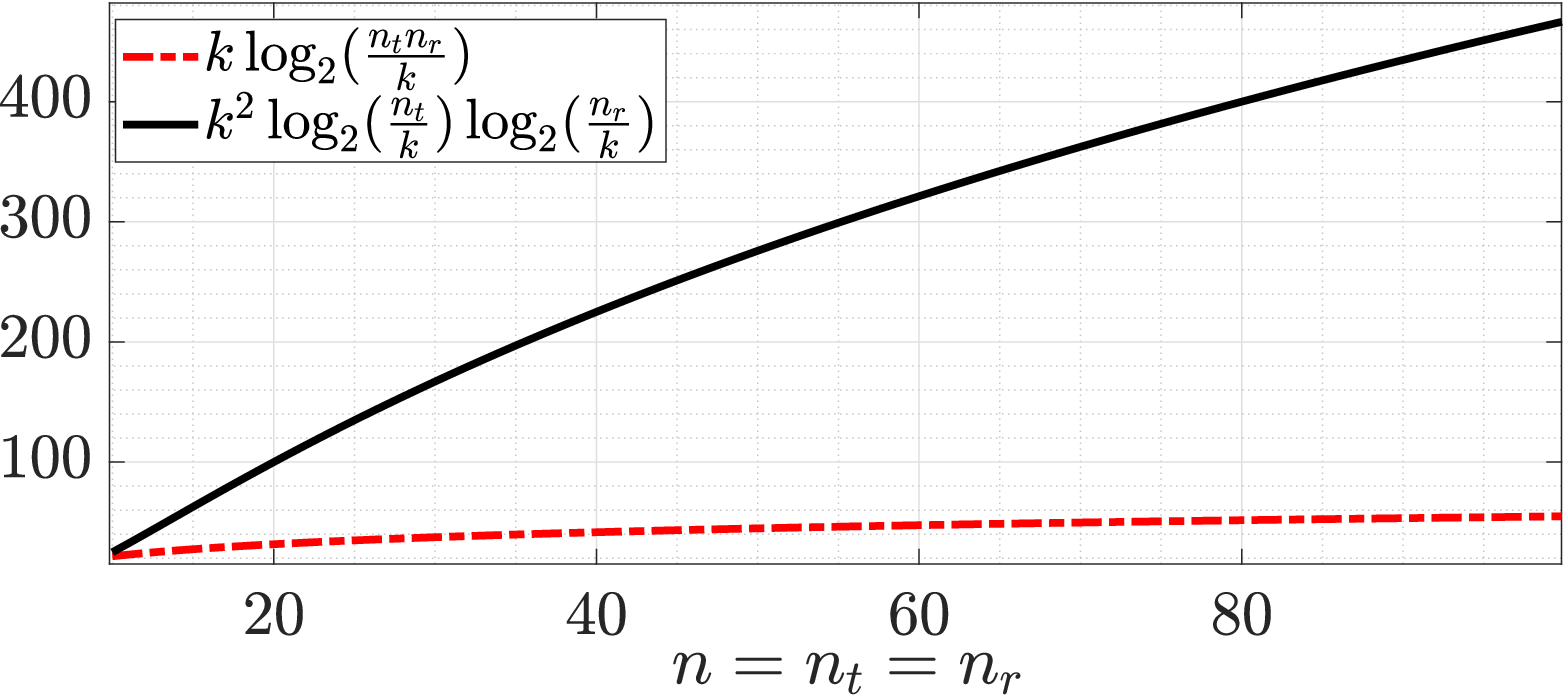}
\caption{At fixed sparsity level $k = 5$.}
\end{subfigure}%
\begin{subfigure}{.49\textwidth}
\centering
\includegraphics[height=3.5cm, width=0.9\linewidth]{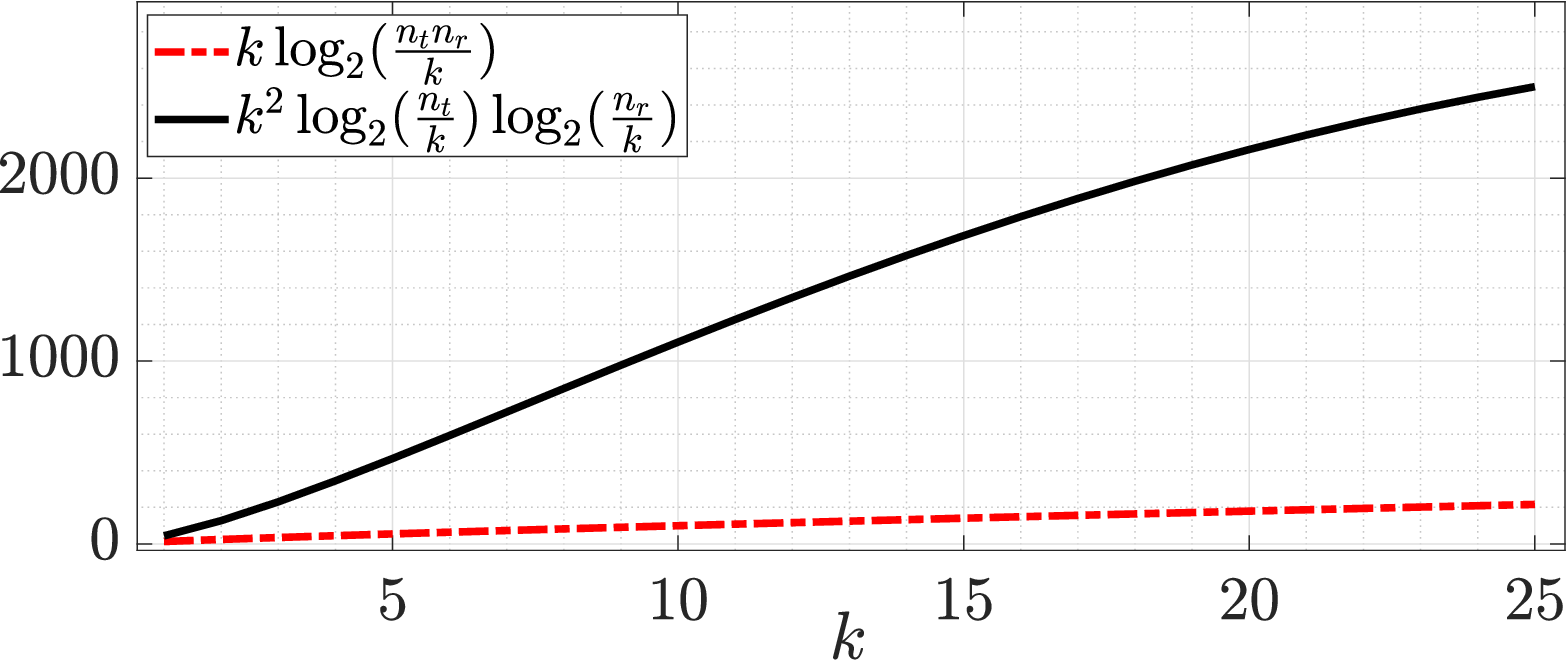}
\caption{At fixed number of antennas $n = n_t = n_r = 100$.}
\end{subfigure}
\caption{Unscaled asymptotic measurement lower bounds.}
\label{fig:LBfuncs}
\end{figure*}
\else
\begin{figure}
\centering
\begin{subfigure}{.49\textwidth}
\centering
\includegraphics[height=2.64cm, width=0.7\linewidth]{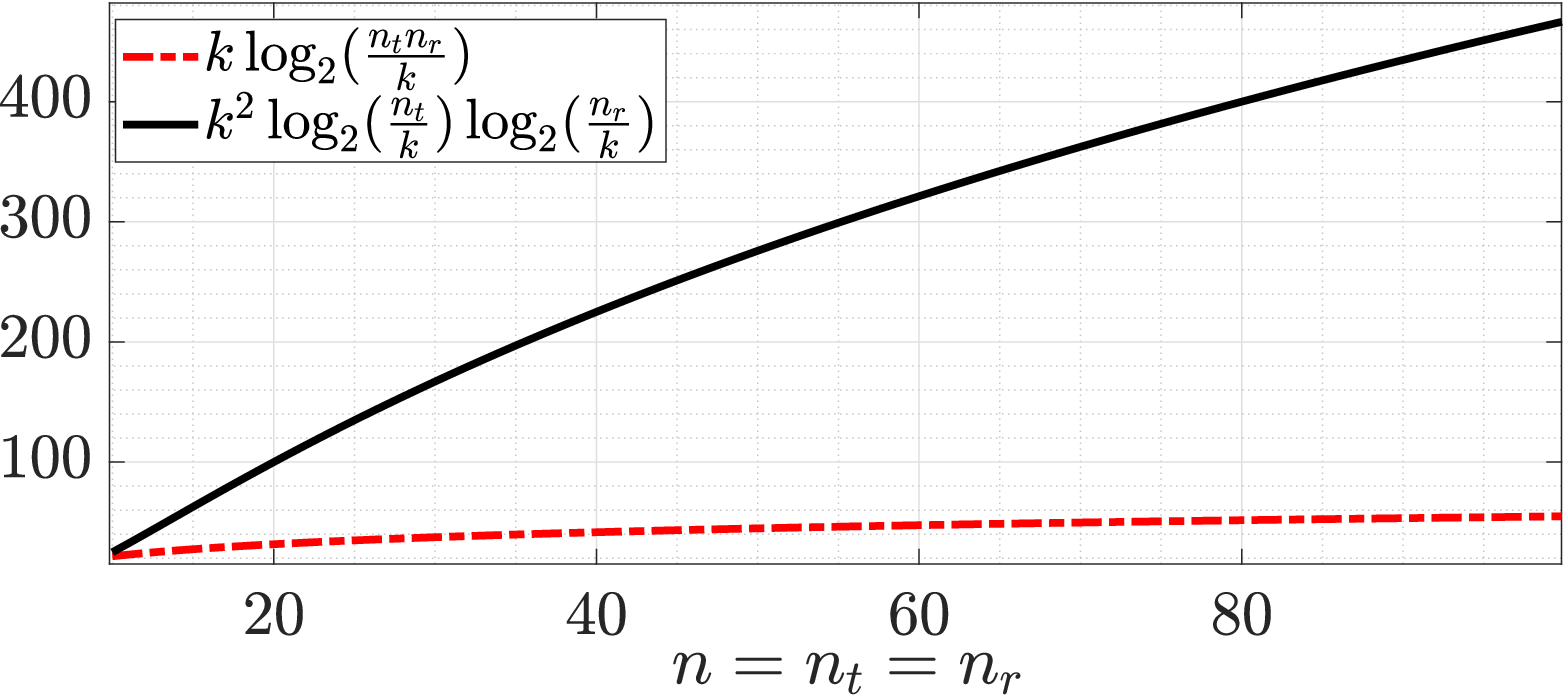}
\caption{At fixed sparsity level $k = 5$.}
\end{subfigure}
\begin{subfigure}{.49\textwidth}
\centering
\includegraphics[height=2.64cm, width=0.7\linewidth]{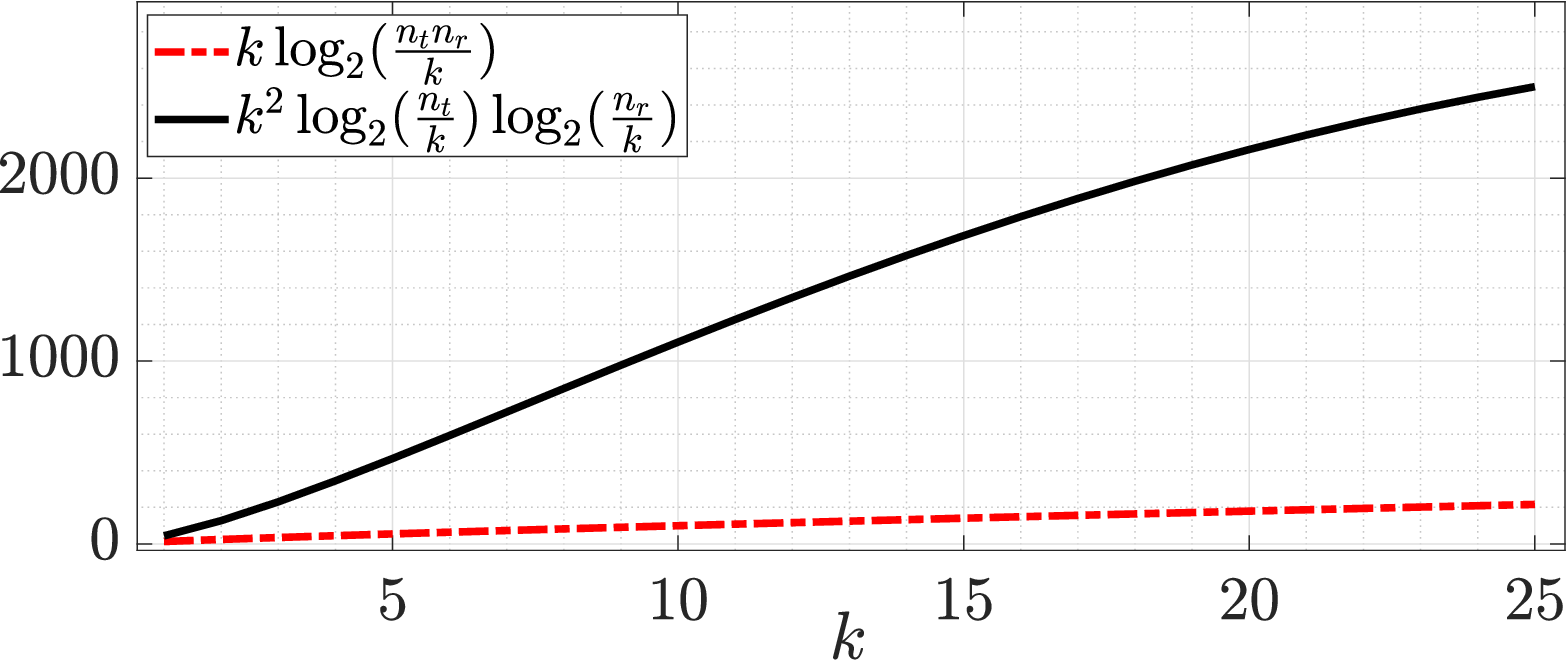}
\caption{At fixed number of antennas $n = n_t = n_r = 100$.}
\end{subfigure}
\caption{Unscaled asymptotic measurement lower bounds.}
\label{fig:LBfuncs}
\end{figure}
\fi

Next, we will derive a tighter bound on the number of measurements. A bound that considers the special structure of the sensing matrix.
This will result in $m = \Omega\left( k^2 \log \left( \frac{n_t}{k}\right) \log \left( \frac{n_r}{k}\right) \right)$.
To appreciate how much tighter our derived bound is, we plot the functions $k \log\left(\frac{n_t \times  n_r}{k}\right)$ and $k^2 \log \left( \frac{n_t}{k}\right) \log \left( \frac{n_r}{k}\right)$ without constant scaling in Fig. \ref{fig:LBfuncs}.

\subsection{Main Results: A ``Tight'' Measurement Bound}
In this section, we will derive the relationship between $k-$RIP constants of Kronecker product matrices and those of the blocks that form it. Then, using Theorem \ref{thm:davenPortBnd}, we will derive an asymptotic lower bound on the number of rows of $\boldsymbol{G_v}$ and deduce its asymptotic behavior. We will finally show the tightness of our derived asymptotic bound using the solution framework in \cite{shabara2019source}.

\textbf{Optimum Measurement Length:}
Among all possible matrices which satisfy the RIP property, we are interested in the ones that have the least number of rows (since the number of rows equals the number of measurements). This leads to the notion of ``\textit{Optimum Measurement Length} (OML)''.
\textit{We define OML as the smallest number of measurements such that the RIP property is satisfied.} OML is dependent on the length of unknown vectors $n$, the maximum sparsity level $k$ and the $k-$RIP constant $\delta$. Hence, we can define a function $\mu$,
\begin{equation}
\mu: \mathcal{N} \times \mathcal{K} \times (0,1) \rightarrow \mathbb{N}_0^+
\end{equation}
which maps the space of all possible values for $n$, $k$, and $\delta$, given by\footnote{We define $\mathbb{N}_0^+$ to be the set of non-negative integers.} $\mathcal{N} \subseteq \mathbb{N}_0^+$, $\mathcal{K}\subseteq \mathbb{N}_0^+$ and $(0,1)$, respectively, to the corresponding OML quantity.

Now, let us focus on the special case of matrices which can be \textit{arbitrarily} constructed. In such case, let $\mu$ be denoted by $\mu_a$ (`a' stands for Arbitrary matrix construction). We define $\mu_a$ to be the solution of the following optimization problem:
\begin{mini!}|l|[2]
{\substack{\boldsymbol{M_a}\in \mathbb{C}^{m_a \times n}}}
         {m_a} {}{P1:}
\addConstraint{ \small \boldsymbol{M_a} \in \mathcal{F}_{\delta}}{}
\end{mini!} 
where $\mathcal{F}_{\delta}$ is the feasible set, and it is defined as
\begin{align*}\small
\mathcal{F}_{\delta} \triangleq \{
\boldsymbol{M_a} \in \mathbb{C}^{m_a \times n} :&
(1 {-} \delta) \norm{\boldsymbol{x}}_2^2 \leq \norm{\boldsymbol{M_a}\boldsymbol{x}}_2^2 \leq (1 {+} \delta) \norm{\boldsymbol{x}}^2_2 , \\ &\forall \boldsymbol{x} \in \mathbb{C}^{n} : \norm{\boldsymbol{x}}_0 \leq k 
\}
\end{align*}
\begin{lemma}\label{lemma:nonDecDelta}
Let $n$ and $k$ be fixed. Then, $\delta_1 \geq \delta_2$ implies $\mu_a(n,k,\delta_1) \leq \mu_a(n,k,\delta_2)$.
\end{lemma}
\ifdefined\longversion
\begin{proof}
\else
\begin{proof*}
\fi
The proof directly follows by observing that $\delta_1 \geq \delta_2$ implies that $\mathcal{F}_{\delta_2} \subseteq \mathcal{F}_{\delta_1}$. Since the problem is a minimization problem, then $\mu_a(n,k,\delta_1) \leq \mu_a(n,k,\delta_2)$.
\ifdefined\longversion
\end{proof}
\else
\end{proof*}
\fi
\textbf{Kronecker Product Matrices:} The standard compressed sensing problem assumes that all elements of the sensing matrix are independently chosen.
On the contrary, in sparse channel estimation, we are restricted to a specific sensing matrix structure, as shown in Eq. (\ref{eq:sensingVect}).
The only free parameters in this sensing matrix are the tx-precoders $\boldsymbol{f_j}$ and the rx-combiners $\boldsymbol{w_i}$.
This limitation suggests that more measurements may be needed to achieve the same RIP constant, compared to matrices whose elements are independently selected.

At the heart of our results lies the relationship between the $k-$RIP constant of Kronecker product matrices and the $k-$RIP constants of the matrices that form them.
We formally state this relationship in the following lemma.
\begin{lemma}[RIP of Kronecker Products]\label{lemma:RIP_kronecker}
Let $\delta_a$ and $\delta_b$ be the $k-$RIP constants of the matrices $\boldsymbol{A}$ and $\boldsymbol{B}$, respectively. Then, the $k-$RIP constant of $\boldsymbol{A} \otimes \boldsymbol{B}$, denoted by $\delta$, is bounded as
\begin{equation}
\delta \geq \max\{\delta_a, \delta_b\}
\end{equation}
\end{lemma}
A similar result to Lemma \ref{lemma:RIP_kronecker} was derived in \cite{jokar2009sparse}, but under the stronger assumption of matrices with normalized columns.
Our more general result implies that even if the normalized columns assumption is loosened, we still cannot obtain a matrix, through a Kronecker Product, which satisfies the RIP property with a constant smaller than the maximum of the $k-$RIP constants of the matrices that form it.
\ifdefined\longversion
The proof of Lemma \ref{lemma:RIP_kronecker} is provided in Appendix \ref{append:RIP_KP}.
\else
To prove Lemma \ref{lemma:RIP_kronecker}, we define two matrices:
$\boldsymbol{C} = \boldsymbol{A} \otimes \boldsymbol{B}$ and $\boldsymbol{C'} = \boldsymbol{B} \otimes \boldsymbol{A}$, whose $k-$RIP constants are $\delta_c$ and $\delta_{c'}$, respectively. We will then show that: (1) $\delta_c \geq \delta_b$, (2) $\delta_{c'} \geq \delta_a$, and that (3) $\delta_c = \delta_{c'}$, by which Lemma \ref{lemma:RIP_kronecker} follows.
A proof outline is provided next, while the detailed proof is provided in \cite{technicalReport}.
To show part (1), let the number of columns of $\boldsymbol{A}$ and $\boldsymbol{B}$ be $n_a$ and $n_b$, respectively. And define $\mathcal{X}_c$ and $\mathcal{X}_b$ to be the sets of vectors with sparsity levels $\leq k$ and whose lengths are $n_a n_b$ and $n_b$, respectively.
Since $\boldsymbol{C}$ and $\boldsymbol{B}$ satisfy $k-$RIP with $\delta_c$ and $\delta_b$. Then, $\norm{\boldsymbol{C} \boldsymbol{x_c}}^2$ is bounded between $(1 \pm \delta_c) \norm{\boldsymbol{x_c}}^2$ for all $\boldsymbol{x_c} \in \mathcal{X}_c$.
Similarly, $\norm{\boldsymbol{B} \boldsymbol{b}}^2$ is bounded between $(1 \pm \delta_b) \norm{\boldsymbol{b}}^2$.
View $\boldsymbol{x_c}$, whose length is $n_a n_b$, as $n_a$ blocks of length $n_b$ each, and focus on a strict subset of $\mathcal{X}_c$, call it $\mathcal{X}_c^{(b)}$, with vectors defined as
\begingroup 
\setlength\arraycolsep{2pt}
$\footnotesize \boldsymbol{x_c^{(b)}} = \small \begin{pmatrix} \boldsymbol{b}^T & \boldsymbol{0}^T & \dots & \boldsymbol{0}^T \end{pmatrix}^T$ for all $\boldsymbol{b} \in \mathcal{X}_b$.
\endgroup
Now examine that $||\boldsymbol{C}\boldsymbol{x_c^{(b)}}||^2 = \norm{\boldsymbol{a_1}}^2 \norm{\boldsymbol{B} \boldsymbol{b}}^2$, where $\boldsymbol{a_1}$ is the first columns of $\boldsymbol{A}$, which leads to these two bounds:
\par\nobreak
\vspace{-9pt}
{\footnotesize
\begin{align}
\label{eq:B1}
&(1-\delta_c)\norm{\boldsymbol{b}}^2 \leq 
\norm{\boldsymbol{a_1}}^2 (1-\delta_b)\norm{\boldsymbol{b}}^2 \leq 
\norm{\boldsymbol{a_1}}^2 \norm{\boldsymbol{B} \boldsymbol{b}}^2 \\
\label{eq:B2}
&\norm{\boldsymbol{a_1}}^2 \norm{\boldsymbol{B} \boldsymbol{b}}^2 \leq \norm{\boldsymbol{a_1}}^2 (1+\delta_b)\norm{\boldsymbol{b}}^2 \leq (1+\delta_c)\norm{\boldsymbol{b}}^2
\end{align}
}%
If $\norm{\boldsymbol{a_1}}\leq 1$, we can use Eq. (\ref{eq:B1}) to conclude that $\delta_c \geq \delta_b$, otherwise, use Eq. (\ref{eq:B2}) to arrive at the same conclusion. This concludes part (1), and by analogy part (2) follows.
Finally, since there exists Permutation matrices $\boldsymbol{P_{\rho}}$ and $\boldsymbol{P_c}$, such that $\boldsymbol{C'} = \boldsymbol{P_{\rho}} \boldsymbol{C} \boldsymbol{P_c}$. This leads to $\norm{\boldsymbol{C'}\boldsymbol{x_c'}} = \norm{\boldsymbol{C}\boldsymbol{x_c}}$, where $\boldsymbol{x_c'} = \boldsymbol{P_c} \boldsymbol{x_c}$ and $\boldsymbol{x_c} \in \mathcal{X}_c$ if and only if $\boldsymbol{x_c'} \in \mathcal{X}_c$. Hence $\delta_c = \delta_{c'}$.
\fi

\textbf{A Generalized Bound:}
Recall Eq. (\ref{eq:sensingVect}). We will rewrite $\boldsymbol{G_v}$, for brevity, in terms of $\boldsymbol{M_t}$ and $\boldsymbol{M_r}$, where 
\ifdefined\longversion
\begin{align}\small
\boldsymbol{M_t} & \triangleq \left( \boldsymbol{F}^H \boldsymbol{U_t} \right)^{\ast} \in \mathbb{C}^{m_t \times n_t} \\
\boldsymbol{M_r} &\triangleq \boldsymbol{W}^H \boldsymbol{U_r} \in \mathbb{C}^{m_r \times n_r} 
\end{align}
\else
\begin{equation*}\small
\boldsymbol{M_t} \triangleq \left( \boldsymbol{F}^H \boldsymbol{U_t} \right)^{\ast} \in \mathbb{C}^{m_t \times n_t},
\quad
\boldsymbol{M_r} \triangleq \boldsymbol{W}^H \boldsymbol{U_r} \in \mathbb{C}^{m_r \times n_r}
\end{equation*}
\fi
Thus, we have $\boldsymbol{G_v} {=} \boldsymbol{M_t} {\otimes} \boldsymbol{M_r}$, and $m {=} m_t m_r$ is the number of rows of $\boldsymbol{G_v}$.
Now, suppose that $\boldsymbol{G_v}$ satisfies $k-$RIP with constant $\delta {\in} (0,1)$. Then, both $\boldsymbol{M_t}$ and $\boldsymbol{M_r}$ must satisfy the $k-$RIP with constants $\delta_t {\in} (0,1)$ and $\delta_r {\in} (0,1)$, respectively.
To show that this is true, assume, without loss of generality (w.l.o.g.), that there does not exist $\delta_t \in (0,1)$ such that $\boldsymbol{M_t}$ satisfies $k-$RIP. Then, there exists a vector $\boldsymbol{v}$ with $\norm{\boldsymbol{v}}_0 {\leq} k$ such that $\boldsymbol{M_t}\boldsymbol{v} = \boldsymbol{0}$, which implies the  existence of at least $k$ dependent columns of $\boldsymbol{M_t}$, call them $\boldsymbol{a_{t1}}, \boldsymbol{a_{t2}}, \dots, \boldsymbol{a_{tk}}$. In turn, there exists at least $k$ dependent columns in $\boldsymbol{G_v}$ (Let $\boldsymbol{a_{r1}}$ be a column in $\boldsymbol{M_r}$, then the columns $\boldsymbol{a_{t1}} {\otimes} \boldsymbol{a_{r1}}, \boldsymbol{a_{t2}} {\otimes} \boldsymbol{a_{r1}}, \dots, \boldsymbol{a_{tk}} {\otimes} \boldsymbol{a_{r1}}$ are dependent).
Hence, $\nexists \delta \in (0,1)$ such that $\boldsymbol{G_v}$ satisfies $k-$RIP with a constant $\delta$. Thus, we arrive at a contradiction.
Further, by Lemma \ref{lemma:RIP_kronecker}, we have that $\delta \geq \max \{\delta_t, \delta_r \}$.

Since $\boldsymbol{M_t}$ and $\boldsymbol{M_r}$ can be arbitrarily constructed, then we can lower bound $m_t$ and $m_r$ by their OML values as follows
\begin{align}
m_t &\geq  \mu_a(n_t, k, \delta_t)
\stackrel{(i)}{\geq} \mu_a(n_t, k, \delta) \\
m_t &\geq \mu_a(n_r, k, \delta_r)
\stackrel{(ii)}{\geq} \mu_a(n_r, k, \delta)
\end{align}
where inequalities $(i)$ and $(ii)$ follow from Lemma \ref{lemma:nonDecDelta}. Thus, it follows that the number of rows of $\boldsymbol{G_v}$, $m$, is bounded as
\begin{equation}\label{eq:generalLB}
m \geq \mu_a(n_t, k, \delta) \times  \mu_a(n_r, k, \delta).
\end{equation}
Recall that $\mu_a(\cdot)$ is the value that solves problem P1.

\begin{remark}
\textit{The implication of Inequality (\ref{eq:generalLB}) is that the number of measurements needed for estimating a sparse MIMO channel, $\boldsymbol{Q}$, is at least equal to (but possibly higher) than the product of the number of measurements needed to solve the following two sub-problems:}
\ifdefined\longversion
\begin{itemize}
\item \textit{The first is a Single-Input Multiple-Output (SIMO), $1 \times n_r$ channel, with $\boldsymbol{M_t}^{\ast}$ as sensing matrix.}
\item \textit{The second is a Multiple-Input Single-Output (MISO), $n_t \times 1$ channel, with $\boldsymbol{M_r}$ as sensing matrix,}
\end{itemize}
\else
\textit{The first is a SIMO, $1 \times n_r$ channel, with $\boldsymbol{M_t}^{\ast}$ as sensing matrix. The second is a MISO, $n_t \times 1$ channel, with $\boldsymbol{M_r}$ as sensing matrix,}
\fi
\textit{where the sparsity level of both channels is $\leq k$.
These two sub-problems are special cases of the original problem, whose measurement equations are shown in
\ifdefined\longversion
Eq. (\ref{eqn:SIMO_measurements}) and Eq. (\ref{eqn:MISO_measurements}), respectively.
\else
Eq. (\ref{eqn:SIMO_MISO_measurements}).
\fi
The only difference is the conjugation of $\boldsymbol{M_t}$.}
\end{remark}

The bound we derive in Eq. (\ref{eq:generalLB}) highlights the dependence on the channel dimensions $n_t$ and $n_r$, the maximum sparsity level $k$ and a measure, $\delta$, of how much information the measurements preserve about the channel. This bound, however, is not explicit, but we can use Theorem \ref{thm:davenPortBnd} to derive a more concrete lower bound for $\mu_a(\cdot)$. This leads to our main result:

\begin{theorem}[Main Theorem]
\label{thm:main_specific}
Fix $\delta \in (0,1)$. If $\boldsymbol{G}_v$ in Eq. (\ref{eq:sensingVect}) satisfies RIP with order $2k$ and constant $\delta$, then the number of measurements $m$ is asymptotically bounded as:
\begin{tcolorbox}[ams equation] \label{eq:mybound}
m = \Omega\left(k^2 \log\left(\frac{n_t}{k}\right) \log\left(\frac{n_r}{k}\right)\right)
\end{tcolorbox}
\end{theorem}
\ifdefined\longversion
\begin{proof}
\else
\begin{proof*}
\fi
Since $\mu_a(n_t, 2k, \delta)$ and $\mu_a(n_r, 2k, \delta)$ are obtained by solving the problem $P1$ (with their respective $n_t$, $n_r$ and $\delta$ values), then there exists matrices $\boldsymbol{X_t}$ and $\boldsymbol{X_r}$, with dimensions $\mu_a(n_t, 2k, \delta) \times n_t$ and $\mu_a(n_r, 2k, \delta) \times n_r$ which satisfy $2k-$RIP with constant $\delta$. Thus, it follows by Theorem \ref{thm:davenPortBnd} that:
\ifdefined\longversion
\begin{align}
\mu_a(n_t, 2k, \delta) & \geq c_{\delta} 2k \log\left( \frac{n_t}{2k} \right)\\
\mu_a(n_r, 2k, \delta) & \geq c_{\delta} 2k \log\left( \frac{n_r}{2k} \right)
\end{align}
\else
\small{
\begin{equation*}
\mu_a(n_t, 2k, \delta)  \geq c_{\delta} 2k \log\left( \frac{n_t}{2k} \right), \; \;
\mu_a(n_r, 2k, \delta)  \geq c_{\delta} 2k \log\left( \frac{n_r}{2k} \right)
\end{equation*}}%
\normalsize
\fi
Therefore, by Eq. (\ref{eq:generalLB}), the following follows
\begin{align}
m = m_t m_r &\geq 4 c_{\delta}^2 k^2 \log\left( \frac{n_t}{2k} \right) \log\left( \frac{n_r}{2k} \right)
\end{align}
Finally, let $c = 0.5$ and recall that the ratio $\frac{n_t}{k}$ increases (by assumption). Then, there exists $n_{t0} \in \mathbb{N}$ such that $\log(\frac{n_t}{2k}) \geq c \log(\frac{n_t}{k})$ for all $n_t \geq n_{t0}$.
Similarly, there exists $n_{r0} \in \mathbb{N}$ such that $\log(\frac{n_r}{2k}) \geq c \log(\frac{n_r}{k})$ for all $n_r \geq n_{r0}$.
Then, it follows that
$m \geq 4 c^2 c_{\delta}^2 k^2 \log\left(\frac{n_t}{k}\right) \log\left(\frac{n_r}{k}\right)$ where $4 c^2 c_{\delta}^2 = c_{\delta}^2$ is a constant, from which Eq. (\ref{eq:mybound}) follows.
\ifdefined\longversion
\end{proof}
\else
\end{proof*}
\fi

\subsection{Tightness of the Measurement Bound}
To argue that the measurement lower bound in Theorem \ref{thm:main_specific} is tight, we will show that there exists a solution, based on \cite{shabara2019source}, which yields sensing matrices that satisfy $2k-$RIP with constants $\in (0,1)$ and with $m \in \Theta \left( k^2 \log \left(\frac{n_t}{k}\right) \log \left(\frac{n_r}{k}\right)\right)$.
We briefly discuss the measurement framework of \cite{shabara2019source} next.

In \cite{shabara2019source}, a source-coding-based framework for the sparse MIMO channel estimation problem is developed.
This solution proposes a method for obtaining a small number of measurements that are sufficient to estimate the channel.
Such measurements are designed based on two carefully chosen binary linear source codes, $C_t$ and $C_r$.
These codes dictate the design of tx-precoders (using $C_t$) and rx-combiners (using $C_r$) and produce real-valued measurement (sensing) matrices, namely, $\boldsymbol{H_t}$ (of size $m_t {\times} n_t$) and $\boldsymbol{H_r}$ (of size $m_r {\times} n_r$), respectively.
The matrix $\boldsymbol{H_t}$ can estimate $k-$sparse MISO channel vectors (i.e., produces unique measurements), while $\boldsymbol{H_r}$ can estimate $k-$sparse SIMO channels.
Hence, the spark of both matrices is greater than $2k$ (by Theorem \ref{thm:spark}).
Measurements are then obtained using all combinations of $m_t$ tx-precoders and $m_r$ rx-combiners, and can be arranged as $\boldsymbol{y_v} = \boldsymbol{H_v} \boldsymbol{q^a_v} + \boldsymbol{n_v}$ where $\boldsymbol{H_v} = \boldsymbol{H_t} \otimes \boldsymbol{H_r}$.
\ifdefined\longversion
By Lemma \ref{lemma:sparkOfTheKroneckerProduct} (in Appendix \ref{append:spark_kron}), we have that $\spark(\boldsymbol{H_v}) > 2k$.
\else
Then, it follows that $\spark(\boldsymbol{H_v}) > 2k$. This is shown in detail in our technical report \cite{technicalReport}.
\fi
Hence, either $\boldsymbol{H_v}$ or a scaled version of it satisfies $2k-$RIP with a constant $\delta_h \in (0,1)$.
This measurement framework is shown to produce a number of measurements, $m$, that is lower bounded as:
\begin{equation} \label{eq:m_underbar} \footnotesize
m \geq \underbar{m} \triangleq \underbrace{\ceil*{\log_2\left(\sum_{i = 0}^{k} {n_r \choose i}\right)}}_{\leq m_t} \underbrace{\ceil*{\log_2\left(\sum_{i = 0}^{k} {n_t \choose i}\right)}}_{\leq m_r}.
\end{equation}
This lower bound is achievable with equality for specific examples as shown in \cite{shabara2019source}.
However, it is not immediately clear how this bound compares to our bound in Eq. (\ref{eq:mybound}).
The following lemma sheds more light on this issue:
\begin{lemma}
The asymptotic behavior of $\underbar{m}$, defined in Eq. (\ref{eq:m_underbar}) follows:
$\underbar{m}  = \Theta \left( k^2 \log \left(\frac{n_t}{k}\right) \log \left(\frac{n_r}{k}\right)\right)$.
\label{lemma:BSC_lb}
\end{lemma}
This is the same asymptotic behavior as the lower bound in Theorem \ref{thm:main_specific}.
\ifdefined\longversion
The proof is provided in Appendix \ref{append:lemma_BSC_lb}.
\else
To prove this lemma, we use the following bound on ${n \choose k}$ \cite{cormen2009introduction}: $\left( \frac{n}{k}\right)^k \leq {n \choose k} \leq \left( \frac{ne}{k}\right)^k$, and observe that for $k < \frac{n+1}{2}$, we have ${n \choose k} \leq \sum_{i = 0}^{k} {n_r \choose i} \leq (k+1) {n \choose k}$.
This allows us to establish upper and lower bounds for $\sum_{i = 0}^{k} {n \choose i}$ using the aforementioned bounds for ${n \choose k}$, which leads to showing that $\log\left(\sum_{i = 0}^{k} {n \choose i}\right) = \Theta \left( k \log \left(\frac{n}{k}\right) \right)$.
Note that $\ceil*{\cdot}$ does not change the asymptotic behavior of its argument.
Therefore, we are able to show that $\underbar{m}  = \Theta \left( k^2 \log \left(\frac{n_t}{k}\right) \log \left(\frac{n_r}{k}\right)\right)$.
A rigorous proof of Lemma \ref{lemma:BSC_lb} is provided in our technical report\cite{technicalReport}.
\fi
Next, we will examine a specific solution based on the family of BCH codes, which results in a number of measurements upper bounded as $ m = O \left( k^2 \log \left(\frac{n_t}{k}\right) \log \left(\frac{n_r}{k}\right)\right)$.

\begin{example}[BCH codes]
Although BCH codes are natively error-correcting codes, they can be used as syndrome-source-codes, as well\footnote{A linear block error-correcting code (LBC) can be utilized as a syndrome source code which can uniquely compress sequences that contain a number of 1's less than or equal to the number of correctable errors of the used code \cite{Ancheta1976Syndrome}. The parity check matrix of the LBC code is used as the generator matrix for the source code. Hence, the number of parity bits of the LBC code is the length of the compressed sequences for the corresponding source code.}.
By the properties of BCH codes, we have that for any positive integers $t \geq 3$ and $k < 2^{t-1}$, there exists a binary BCH code with: i) block length $n = 2^t-1$, ii) minimum distance $d_{\text{min}} \geq 2k+1$ (hence, it can correct up to $k$ errors), and iii) a number of parity check bits $m \leq t k = k \log_2\left(n+1\right)$.
Using BCH codes to design $C_t$ and $C_r$, we obtain a solution whose number of measurements is upper bounded according to the following lemma:
\begin{lemma} \label{lemma:measurement_upperbound}
The number of measurements achievable using BCH codes in the framework of \cite{shabara2019source} is asymptotically bounded as $m = O \left( k^2 \log \left(\frac{n_t}{k}\right) \log \left(\frac{n_r}{k}\right) \right)$.
\end{lemma}
\ifdefined\longversion
The proof depends on constructing syndrome source codes with arbitrary block lengths, and is provided in Appendix \ref{append:lemma_proof_upperBound}.
\else
\begin{proof*}[Proof Sketch]
For arbitrary values of $n_t {\geq} 7$, there exists $t$ such that $n_t {\leq} 2^{t+1} {\triangleq} n_t'$.
Then, for all $k {<} \frac{n_t'+1}{2}$, there exists a BCH code with block length $n_t'$ and parity length of $m_t' {=} O(k\log n_t')$.
We can then shorten that BCH code by removing $n_t' {-} n_t$ information bits from its codewords while keeping $m_t {=} m_t'$ unchanged. We can then show that $m_t {=} m_t' {=} O(k\log n_t)$, since $\log 2^{t+1} {\leq} \frac{4}{3}\log 2^t$.
Now, recall that $n_t \geq k^{1+\epsilon}$, where $\epsilon > 0$ (by assumption).
Then, $\frac{1}{\epsilon} \log \frac{n_t}{k} {\geq} \log k$. Therefore, we can show that $\log(n_t) {\leq} (1+\frac{1}{\epsilon})\log\left(\frac{n_t}{k}\right)$, by which we have $m_t {=} O(k\log\left(\frac{n_t}{k})\right)$.
Similarly, $m_r {=} O(k \log\left(\frac{n_r}{k})\right)$.
Thus, $m = O \left( k^2 \log \left(\frac{n_t}{k}\right) \log \left(\frac{n_r}{k}\right) \right)$. For a detailed proof, see our technical report \cite{technicalReport}.
\end{proof*}
\fi
\end{example}

Among all solutions in \cite{shabara2019source}, we are interested in the ones whose number of measurements, $m$, is closest to $\underbar{m}$. These solutions are \textit{``Optimum''} in the sense of reducing the number of measurements.
Recall that $\underbar{m}$ is the lower bound of all solutions based on \cite{shabara2019source} (see Eq. (\ref{eq:m_underbar})).
The following theorem shows that these optimum solutions scale similarly to $\underbar{m}$, which in turn shows that the lower bound of Theorem \ref{thm:main_specific} is tight.
\color{black}
\begin{theorem}
The number of measurements of ``Optimum Solutions'' of \cite{shabara2019source} scales as
$m = \Theta \left( k^2 \log \left(\frac{n_t}{k}\right) \log \left(\frac{n_r}{k}\right)\right)$
\end{theorem}
\ifdefined\longversion
\begin{proof}
\else
\begin{proof*}
\fi
By Lemma \ref{lemma:BSC_lb}, we have that all solutions, including the optimal, have $m = \Omega \left( k^2 \log \left(\frac{n_t}{k}\right) \log \left(\frac{n_r}{k}\right) \right)$.
Moreover, Lemma \ref{lemma:measurement_upperbound} shows that solutions based on BCH codes result in $m = O \left( k^2 \log \left(\frac{n_t}{k}\right) \log \left(\frac{n_r}{k}\right)\right)$. Since optimal solutions have a number of measurements smaller than or equal to those obtained by BCH codes, then they also have the same asymptotic upper bound. Therefore, optimal solutions have $m = \Theta \left( k^2 \log \left(\frac{n_t}{k}\right) \log \left(\frac{n_r}{k}\right)\right)$ follows.
\ifdefined\longversion
\end{proof}
\else
\end{proof*}
\fi

\begin{remark}
Even though we have shown that the bound of Theorem \ref{thm:main_specific} is tight, we have demonstrated this tightness in the asymptotic regime of $n$ and $k$.
The dependence on the RIP constant, $\delta$, however, remains an open question.
\end{remark}

%% file: text/conclusion.tex
In this paper, we study the fundamental lower bound governing the number of measurements, required for estimating sparse, large MIMO channels.
We consider a simple analog transceiver, where each channel measurements is obtained using a specific combination of beamforming vectors at the transmitter and receiver.
The currently known lower bound on number of measurements is $\Omega\left(k\log\left(\frac{n_r n_t}{k}\right) \right)$.
We derive a tight lower measurement bound, which scales asymptotically as $\Omega\left(k^2\log\left(\frac{n_t}{k}\right) \log\left(\frac{n_r}{k}\right) \right)$. 
The tightness of our derived bound is demonstrated by showing that there exists a solution with $m = O\left(k^2\log\left(\frac{n_t}{k}\right) \log\left(\frac{n_r}{k}\right) \right)$.

%% file: text/RIP_KP.tex
\begin{proof}
Let $\boldsymbol{A} \in \mathbb{R}^{m_a \times n_a}$ and $\boldsymbol{B} \in \mathbb{R}^{m_b \times n_b}$. Denote by $\boldsymbol{a_i}$ the $i^{\text{th}}$ column of $\boldsymbol{A}$ and let $a_{i,j}$ be its $j^{\text{th}}$ element. And define $\boldsymbol{C} \triangleq \boldsymbol{A} \otimes \boldsymbol{B}$. Denote by $\delta_c$ the $k-$RIP constant of $\boldsymbol{C}$.

We will first show that $\delta_c \geq \delta_b$.
To that end, let us define the sets $\mathcal{X}_c$ and $\mathcal{X}_b$ as:
\begin{align}
\mathcal{X}_c &\triangleq \{ \boldsymbol{x_c} \in  \mathbb{R}^{n_a n_b} : \norm{\boldsymbol{x_c}}_0 \leq k \}\\
\mathcal{X}_b &\triangleq \{ \boldsymbol{x_b} \in  \mathbb{R}^{n_b} : \norm{\boldsymbol{x_b}}_0 \leq k \}
\end{align}
Since $\delta_c$ is the $k-$RIP constant of $\boldsymbol{C}$, then $\forall \boldsymbol{x_c} \in \mathcal{X}_c$ we have
\begin{equation}\label{eq:X_c}
(1-\delta_c)\norm{\boldsymbol{x_c}}^2 \leq  \norm{\boldsymbol{C}\boldsymbol{x_c}}^2 \leq (1+\delta_c) \norm{\boldsymbol{x_c}}^2
\end{equation}
Now, we will focus our attention on a smaller class of vectors $\boldsymbol{x}^{(\boldsymbol{b})}_{\boldsymbol{c}}$, which constitute a strict subset of $\mathcal{X}_c$, defined as follows
\ifdefined\longversion
\begin{equation}
\boldsymbol{x}^{(\boldsymbol{b})}_{\boldsymbol{c}} \triangleq
\begin{pmatrix}
\boldsymbol{b} \\
\boldsymbol{0} \\
\vdots \\
\boldsymbol{0}
\end{pmatrix},
\end{equation}
\else
\begin{equation}
\boldsymbol{x}^{(\boldsymbol{b})}_{\boldsymbol{c}} \triangleq
\begin{pmatrix}
\boldsymbol{b}^T &
\boldsymbol{0}^T &
\dots            &
\boldsymbol{0}^T
\end{pmatrix}^T,
\end{equation}
\fi
where $\boldsymbol{b} \in \mathcal{X}_b$ and $\boldsymbol{x}^{(\boldsymbol{b})}_{\boldsymbol{c}} \in \mathbb{R}^{n_a n_b}$.
Then, by construction, $\boldsymbol{x}^{(\boldsymbol{b})}_{\boldsymbol{c}} \in \mathcal{X}_c$, and $\norm{\boldsymbol{x}^{(\boldsymbol{b})}_{\boldsymbol{c}}} = \norm{\boldsymbol{b}}$.
Now, observe that $\small \norm{\boldsymbol{C}\boldsymbol{x}^{(\boldsymbol{b})}_{\boldsymbol{c}}}^2$ is
\begin{align}\small
\norm{\boldsymbol{C}\boldsymbol{x}^{(\boldsymbol{b})}_{\boldsymbol{c}}}^2 =
\norm{\left(\boldsymbol{A} \otimes  \boldsymbol{B}\right)\boldsymbol{x}^{(\boldsymbol{b})}_{\boldsymbol{c}}}^2
&= \sum_{i=1}^{n_a} |a_{i,1}|^2 \norm{\boldsymbol{B}\boldsymbol{b}}^2  \\
&= \norm{\boldsymbol{a_1}}^2 \norm{\boldsymbol{B}\boldsymbol{b}}^2
\end{align}
Since $\delta_b$ is the $k-$RIP constant of $\boldsymbol{B}$, then $\forall \boldsymbol{b} \in \mathcal{X}_b$ we have
\begin{equation} \label{eq:db} \small
\norm{\boldsymbol{a_1}}^2 (1-\delta_b) \norm{\boldsymbol{b}}^2
\leq \underbrace{\norm{\boldsymbol{a_1}}^2 \norm{\boldsymbol{B}\boldsymbol{b}}^2}_{= \norm{\boldsymbol{C}\boldsymbol{x}^{(\boldsymbol{b})}_{\boldsymbol{c}}}^2}
\leq \norm{\boldsymbol{a_1}}^2 (1+\delta_b) \norm{\boldsymbol{b}}^2
\end{equation}
Since (i) the space of all possible constructions of $\boldsymbol{x}^{(\boldsymbol{b})}_{\boldsymbol{c}}$ is a strict subset of $\mathcal{X}_c$, and since (ii) $\delta_b$ is the smallest constant such that Eq. (\ref{eq:db}) holds, then the following two equations must always hold true
\begin{align}
(1-\delta_c) \norm{\boldsymbol{b}}^2 \leq &\norm{\boldsymbol{a_1}}^2 (1-\delta_b) \norm{\boldsymbol{b}}^2 \tag{B1} \label{eq:B1}\\
&\norm{\boldsymbol{a_1}}^2 (1+\delta_b) \norm{\boldsymbol{b}}^2 \leq (1+\delta_c) \norm{\boldsymbol{b}}^2 \tag{B2} \label{eq:B2}
\end{align}
If $\norm{\boldsymbol{a_1}}^2 \leq 1$, then from Eq. (\ref{eq:B1}) we have $\delta_c \geq \delta_b$. Otherwise, if $\norm{\boldsymbol{a_1}}^2 \geq 1$, then from Eq. (\ref{eq:B2}) we have $\delta_c \geq \delta_b$.
Therefore, $\delta_b \leq \delta_c$ is always true.
Now, define $C' \triangleq \boldsymbol{B} \otimes \boldsymbol{A}$.
By the properties of the Kronecker product, we know that there exist two ``Permutation'' matrices, call them $\boldsymbol{P_{\rho}}$ and $\boldsymbol{P_c}$, such that:
\begin{equation}
\boldsymbol{C'} =
\boldsymbol{P_{\rho}} \boldsymbol{C} \boldsymbol{P_c} = 
\boldsymbol{P_{\rho}} \left(\boldsymbol{A} \otimes \boldsymbol{B}\right) \boldsymbol{P_c},
\end{equation}
where $\boldsymbol{P_{\rho}}$ permutes the rows of $\boldsymbol{C}$, and $\boldsymbol{P_c}$ permutes the columns of $\boldsymbol{P_{\rho}} \boldsymbol{C}$.
Then, we have that
\begin{align}
\norm{\boldsymbol{C'} \boldsymbol{x_c}}
= \norm{\boldsymbol{P_{\rho}} \boldsymbol{C} \boldsymbol{P_c} \boldsymbol{x_c}}
\stackrel{(i)}{=} \norm{\boldsymbol{C} \boldsymbol{P_c} \boldsymbol{x_c}}.
\end{align}
Also, observe that if $\boldsymbol{x_c} \in \mathcal{X}_c$, then $\boldsymbol{P_c}\boldsymbol{x_c}$ has the same sparsity level as $\boldsymbol{x_c}$ and hence it lies in $\mathcal{X}_c$, as well. Therefore, it follows that
\begin{equation}
(1-\delta_c) \norm{\boldsymbol{x_c}}^2
\leq \underbrace{\norm{\boldsymbol{C'} \boldsymbol{x_c}}^2}_{= \norm{\boldsymbol{C} \boldsymbol{P_c} \boldsymbol{x_c}}^2}
\leq (1+\delta_c) \norm{\boldsymbol{x_c}}^2,
\end{equation}
which shows that both $\boldsymbol{C}$ and $\boldsymbol{C'}$ have the same $k-$RIP constant $\delta_c$. Then, it follows that $\delta_c \geq \delta _a$.
Therefore, $\delta_c \geq \max\{\delta_a, \delta_b\}$, which concludes our proof.
\end{proof}

%% file: text/Appendix_V.tex
\begin{proof}
First, observe that ${n \choose i} < {n \choose k}$ for all $0 \leq i < k$ such that $k < \frac{n+1}{2}$.
Thus, for $k < \frac{n+1}{2}$, we have that
\begin{equation}
\sum_{i = 0}^k {n \choose i} \leq \left( k+1\right) {n\choose k} \footnotesize
\end{equation}
By taking the logarithm of the previous equation, we get
\small{
\begin{align}
\Rightarrow \log \left( \sum_{i = 0}^k {n \choose i} \right) &\leq \log \left( k+1\right) + \log {n \choose k} \footnotesize \\
&\leq \log \left( k+1\right) + k \log  \left( \frac{n}{k}\right) + k \log e, \label{eq:SC_upperBound}
\end{align}}%
\normalsize
where (\ref{eq:SC_upperBound}) follows from the following popular bounds on $n \choose k$ \cite{cormen2009introduction}
\begin{equation}\footnotesize
\left( \frac{n}{k}\right)^k \leq {n \choose k} \leq \left( \frac{ne}{k}\right)^k. \label{eq:nchoosekBounds}
\end{equation}
\normalsize
From (\ref{eq:nchoosekBounds}) we also have that $\left( \frac{n}{k}\right)^k \leq {n \choose k} \leq \sum_{i = 0}^k {n \choose i}$. This gives us the following upper and lower bounds on $\Upsilon$ where
\begin{equation}\footnotesize
\Upsilon \triangleq \ceil*{\log_2\left(\sum_{i = 0}^{k} {n \choose i}\right)}
\end{equation}
\begin{equation}\small
k \log \left( \frac{n}{k}\right) \leq \Upsilon \leq \log (k+1) + k \log \left( \frac{n}{k}\right) + k \log e +1
\end{equation}
Therefore, we have that $\Upsilon$ is in both $\Omega\left(k\log\frac{n}{k}\right)$ and $O\left(k\log\frac{n}{k}\right)$.
Hence, $\Upsilon \in \Theta\left(k\log\frac{n}{k}\right)$.
Finally, we can conclude that $\underbar{m} = \Upsilon\vert_{n = n_t} \Upsilon\vert_{n = n_r}$ is asymptotically bounded as
\begin{equation*}
\underbar{m} = \Theta\left( k^2 \log \left( \frac{n_t}{k}\right) \log \left( \frac{n_r}{k}\right) \right) \qedhere
\end{equation*}
\end{proof}

%% file: text/Appendix_III.tex
\normalsize
\begin{proof}
First, we will show that $m_t \leq c k \log \frac{n_t}{k}$ for some $c >0 $. Let $n_t$ be an arbitrary integer such that $n_t >=7$.
Then, there exists a positive integer $t\geq3$ such that
$2^{t}-1 \leq n_t < 2^{t+1}-1$.
If $n_t = 2^{t}-1$, then there exists a BCH code with a number of parity check bits $m_t$ such that
$m_t \leq k \log(n_t +1)$. Hence, there exists a positive constant $c_0 \in \mathbb{R}$ such that $m_t \leq c_0 k \log(n_t)$.
On the other hand, if $2^{t}-1 < n_t < 2^{t+1}-1$, then we can construct a linear block code of length $n_t$ by shortening a BCH code with block length $n_t' = 2^{t+1}-1$ and number of parity check bits $m_t \leq k \log(n_t' + 1)$. This shortening process leaves the number of parity bits intact, hence, we have that $m_t \leq k \log\left( n_t'+1 \right)$, but it removes $n_t' - n_t$ information bits from the codewords. Thus, we have
\small{
\begin{equation}
m_t \leq k \log\left( 2^{t+1} \right) \leq \frac{4}{3} k \log\left( 2^t \right) < \frac{4}{3} k \log\left( n_t + 1 \right).
\end{equation}}%

\normalsize
Now, recall that $n_t \geq k^{1+\epsilon}$, where $\epsilon > 0$ (by assumption).
Then, we have that $\frac{1}{\epsilon} \log \frac{n_t}{k} \geq \log k$. Therefore,
\small{
\begin{align}
\log(n_t)
&=            \log\left( \frac{n_t}{k} \times k \right)
=     \left( \log\left( \frac{n_t}{k} \right) + \log \left( k \right) \right)\\
&\leq  \left( \log\left( \frac{n_t}{k} \right) + \frac{1}{\epsilon} \log \frac{n_t}{k}\right)
=     \left( 1 + \frac{1}{\epsilon}\right) \log\left( \frac{n_t}{k} \right)
\end{align}}%
\normalsize
Thus, if follows that $m_t = O(k\log\left( \frac{n_t}{k} \right))$ for arbitrary $n_t \in \mathbb{N}$.\\
Similarly, we have $m_r = O\left( k \log \left( \frac{n_r}{k} \right) \right)$.
Thus, it follows that $m = O(k^2 \log\left( \frac{n_t}{k} \right) \log\left( \frac{n_r}{k} \right))$.
\end{proof}

%% file: text/spark_kron.tex
\begin{lemma}\label{lemma:sparkOfTheKroneckerProduct}
Let $\boldsymbol{A} \in \mathbb{C}^{m_a \times n_a}$ and $\boldsymbol{B} \in \mathbb{C}^{m_b \times n_b}$ be such that $\min\{\spark(\boldsymbol{A}), \spark(\boldsymbol{B}) \} > k$. Then, $\spark(\boldsymbol{A} \otimes \boldsymbol{B}) > k$.
\end{lemma}
\begin{proof}
Let $\left( \boldsymbol{a_i} \right)_{i = 1}^{n_a}$  and $\left( \boldsymbol{b_j} \right)_{j = 1}^{n_b}$ be the columns of $\boldsymbol{A}$ and $\boldsymbol{B}$, respectively.
Since $\spark(\boldsymbol{A}) > k$, then any $k$ columns of $\boldsymbol{A}$ are linearly independent. Similarly, any $k$ columns of $\boldsymbol{B}$ are also independent.
Observe that any column of $\boldsymbol{A} \otimes \boldsymbol{B}$ is of the form $\boldsymbol{a_i} \otimes \boldsymbol{b_j}$. Pick any $k$ columns of $\boldsymbol{A} \otimes \boldsymbol{B}$, i.e., $\boldsymbol{a_{p_1}} \otimes \boldsymbol{b_{t_1}}$, $\boldsymbol{a_{p_2}} \otimes \boldsymbol{b_{t_2}}$, $\dots$, $\boldsymbol{a_{p_k}} \otimes \boldsymbol{b_{t_k}}$.
We will show that $\sum_{i = 1}^k \alpha_i \boldsymbol{a_{p_i}} \otimes \boldsymbol{b_{t_i}} = \boldsymbol{0}$ if and only if $\alpha_i = 0 \; \forall i$.

\ifdefined\longversion
Assume, without loss of generality, that
\begin{align*}\footnotesize
\boldsymbol{a_{p_1}} &= \dots = \boldsymbol{a_{p_{d_1}}}, \text{  and}
  \quad \sum_{i = 1}^{d_1} \alpha_i \boldsymbol{b_{t_i}} = \boldsymbol{r_{d_1}} \\
\boldsymbol{a_{p_{d_1+1}}} &= \dots = \boldsymbol{a_{p_{d_2}}}, \text{  and}
  \quad \sum_{i = d_1+1}^{d_2} \alpha_i \boldsymbol{b_{t_i}} = \boldsymbol{r_{d_2}}\\
&\quad \quad \vdots \hspace{3cm}\vdots \\
\boldsymbol{a_{p_{d_{l-1}+1}}} &= \dots = \boldsymbol{a_{p_{d_l}}}, \text{  and}
  \quad \sum_{i = d_{l-1}+1}^{d_l = k} \alpha_i \boldsymbol{b_{t_i}} = \boldsymbol{r_{d_l}}\\
\end{align*}
Then, we can rewrite $\sum_{i = 1}^k \alpha_i \boldsymbol{a_{p_i}} \otimes \boldsymbol{b_{t_i}}$ as:
\begin{align*}\tiny
&\boldsymbol{a_{p_{d_1}}} \otimes \underbrace{\left( \sum_{i = 1}^{d_1} \alpha_i \boldsymbol{b_{t_i}}\right)}_{\boldsymbol{r_{d_1}}} + 
 \boldsymbol{a_{p_{d_2}}} \otimes \underbrace{\left( \sum_{i = d_1+1}^{d_2} \alpha_i \boldsymbol{b_{t_i}} \right)}_{\boldsymbol{r_{d_2}}} \\
& \hspace{3cm} + \dots + \boldsymbol{a_{p_{d_l}}} \otimes \underbrace{\left( \sum_{i = d_{l-1}+1}^{d_l} \alpha_i \boldsymbol{b_{t_i}} \right)}_{\boldsymbol{r_{d_l}}}
\end{align*}
\else
Assume w.l.o.g. that
$\small \boldsymbol{a_{p_1}}           {=} \dots {=} \boldsymbol{a_{p_{d_1}}}$,
$\small \boldsymbol{a_{p_{d_1+1}}}     {=} \dots {=} \boldsymbol{a_{p_{d_2}}}$, $\dots$,
$\small \boldsymbol{a_{p_{d_{l-1}+1}}} {=} \dots {=} \boldsymbol{a_{p_{d_l}}}$.
And define
$\small \boldsymbol{r_{d_1}} {\triangleq} \sum_{i = 1}^{d_1} \alpha_i \boldsymbol{b_{t_i}}$,
$\small \boldsymbol{r_{d_2}} {\triangleq} \sum_{i = d_1+1}^{d_2} \alpha_i \boldsymbol{b_{t_i}}$,
$\dots$,
$\small \boldsymbol{r_{d_l}} {\triangleq} \sum_{i = d_{l-1}+1}^{d_l = k} \alpha_i \boldsymbol{b_{t_i}}$.
Then, we can rewrite $\sum_{i = 1}^k \alpha_i \boldsymbol{a_{p_i}} \otimes \boldsymbol{b_{t_i}}$ as:
$\boldsymbol{a_{p{d_1}}}  \otimes \boldsymbol{r_{d_1}} + 
\boldsymbol{a_{p{d_2}}}  \otimes \boldsymbol{r_{d_2}} + \dots +
\boldsymbol{a_{p{d_l}}}  \otimes \boldsymbol{r_{d_l}}$.
\fi
Suppose there exists at least one value $i_0$ such that $\alpha_{i_0} \neq 0$, then $\exists \boldsymbol{r_{d_i}} \neq 0$ since all $\boldsymbol{b_{t_i}}$ are independent. Finally, since all $\boldsymbol{a_{p{d_i}}}$ are independent,
then $\sum_{i = 1}^k \alpha_i \boldsymbol{a_{p_i}} \otimes \boldsymbol{b_{t_i}} \neq \boldsymbol{0}$.
Therefore, the $k$ columns $\left( \boldsymbol{a_{p_i}} \otimes \boldsymbol{b_{t_i}} \right)_{i = 1}^k$, of $\boldsymbol{A} \otimes \boldsymbol{B}$, are independent. Hence, $\spark(\boldsymbol{A} \otimes \boldsymbol{B}) > k$.
\end{proof}